\documentclass{revtex4}
\usepackage{amsmath,amsthm}

\usepackage{graphicx}

\usepackage{amsfonts}
\newtheorem{theorem}{Theorem}[section]

\newtheorem{proposition}[theorem]{Proposition}

\theoremstyle{remark}
\newtheorem{remark}[theorem]{Remark}

\theoremstyle{definition}
\newtheorem{definition}[theorem]{Definition}

\theoremstyle{example}

\theoremstyle{notation}

\newcommand{\bra}[1]{\langle#1|}
\newcommand{\ket}[1]{|#1\rangle}

\begin{document}
\draft
\title{Weak mutually unbiased bases}
\author{M. Shalaby, A. Vourdas\\
Department of Computing,\\
University of Bradford, \\
Bradford BD7 1DP, United Kingdom}

\begin{abstract}
Quantum systems with variables in ${\mathbb Z}(d)$ are considered.
The properties of lines in the ${\mathbb Z}(d)\times {\mathbb Z}(d)$ phase space of these systems, are studied.
Weak mutually unbiased bases in these systems are defined as bases for which the 
overlap of any two vectors in two different bases, is equal to $d^{-1/2}$ or alternatively to one of the $d_i^{-1/2},0$
(where $d_i$ is a divisor of $d$ apart from $d,1$).
They are designed for the geometry of the ${\mathbb Z}(d)\times {\mathbb Z}(d)$ phase space, in the sense
that there is a duality between the weak mutually unbiased bases and the maximal lines through the origin. 
In the special case of prime $d$, there are no divisors of $d$ apart from $1,d$ and the weak mutually 
unbiased bases are mutually unbiased bases.
\end{abstract}
{\large PUBLISHED IN: J. Phys. A45, 052001 (2012)}

\maketitle

\section{Introduction}
There is much work on quantum systems, 
where the position and momentum take values 
in the ring ${\mathbb Z}(d)$ (the integers modulo $d$).
Recent reviews have been presented in \cite{1,2,3,4,5}.

An important topic in this context is the mutually unbiased bases 
\cite{m1,m2,m3,m4,m5,m7,m8,m9,m10,m11,m12,m13,m14}. 
For a prime number $d$, the number of mutually unbiased bases is equal to $d+1$.
These results can be extended to
quantum systems where the position and momentum take values in a Galois fields $GF(p^n)$ \cite{Gal1,Gal2} 
(where $p$ is a prime number).
Hamiltonians for the construction of such systems from $n$ component subsystems, each of which is $p$-dimensional, 
have been discussed 
in \cite{V}.
The number of mutually unbiased bases in these systems is equal to $p^n+1$.
There is currently much work on mutually unbiased bases in systems with dimension which is not a power of a 
prime number
(e.g., for the case $d=6$\cite{stefan1,stefan2}).

In mutually unbiased bases the absolute value of the overlap of any two vectors in two different bases is 
$d^{-1/2}$.
Let $d_i$ be the divisors of $d$, apart from $d$ and $1$.
In this paper we introduce the concept of weak mutually unbiased bases, where roughly speaking, 
this absolute value is equal to $d^{-1/2}$ or alternatively to one of the $d_i^{-1/2},0$
(the precise definition is given below).
We show a correspondence (`duality') between the properties of the weak mutually unbiased bases and the 
properties of the maximal lines in the
${\mathbb Z}(d)\times {\mathbb Z}(d)$ phase space of these systems.
This duality shows that our concept of weak mutually unbiased bases is intimately related and motivated 
by the geometrical properties
of the phase space of these systems.
We note here that the term  weak mutually unbiased bases, has also been used in \cite{weak} 
for a different concept.

The ${\mathbb Z}(d)\times {\mathbb Z}(d)$ phase space is a finite geometry \cite{f1,f2,f3}.
Most of the existing work on finite geometries is on near-linear geometries where 
two points belong to at most one line.
In \cite{SV} we have pointed out that  ${\mathbb Z}(d)\times {\mathbb Z}(d)$ violates this axiom and
that it is not a near-linear geometry.
Only in the special case that $d$ is a prime number and ${\mathbb Z}(d)$ is a field, the 
${\mathbb Z}(d)\times {\mathbb Z}(d)$ is a near-linear geometry (the $GF(p^n)\times GF(p^n)$ is also a near-linear geometry).

In section II we study lines in ${\mathbb Z}(d)\times {\mathbb Z}(d)$.
In particular we show how symplectic transformations map the various lines into other lines.
We also discuss a factorization of lines in ${\mathbb Z}(d)\times {\mathbb Z}(d)$ 
in terms of `component lines' in ${\mathbb Z}(d_i)\times {\mathbb Z}(d_i)$
where $d=\prod d_i$ and any $(d_i,d_j)$ are coprime.
This section extends considerably our previous results for lines in ${\mathbb Z}(d)\times {\mathbb Z}(d)$ in ref.\cite{SV}.

In section III we discuss briefly quantum systems with variables in ${\mathbb Z}(d)$ in order to establish the notation.
In section IV we define weak mutually unbiased bases.
For simplicity we consider the case with $d=d_1d_2$, where $d_1$ and $d_2$ are prime numbers 
(where $d_1<d_2$).
In this case the divisors of $d$ (apart from $1,d$) are $d_1, d_2$.
This leads to bases which are tensor products of mutually unbiased bases in the Hilbert spaces
$H_{d_1}$ and $H_{d_2}$, describing systems with variables in
${\mathbb Z}(d_1)$ and ${\mathbb Z}(d_2)$, correspondingly.
We have considered recently such bases in the context of tomography \cite{SV}, 
but here we arrive at these bases starting from the above definition for
weak mutually unbiased bases.
We note that a tensor product of  mutually unbiased bases in two Hilbert spaces has also been used in \cite{K}.
In section V we discuss the duality between 
the weak mutually unbiased bases and the lines in the
${\mathbb Z}(d)\times {\mathbb Z}(d)$ phase space of these systems.
We conclude in section VI with a discussion of our results.

\section{Lines in ${\mathbb Z}(d)\times {\mathbb Z}(d)$}

If $d_i$ is a divisor of $d$, then $(d/d_i){\mathbb Z}(d_i)$ is the subgroup of ${\mathbb Z}(d)$
which consists of the elements $0,d/d_i,2d/d_i,...,[(d_i-1)d/d_i]$.
We factorize $d$ as
\begin{eqnarray}
d=\prod _{i=1}^Np_i^{e_i},
\end{eqnarray} 
where $p_i$ are prime numbers.
The Dedekind function is given by
\begin{eqnarray}
\psi(d)=d\prod _{i=1}^N\left (1+\frac{1}{p_i}\right ).
\end{eqnarray} 
The Jordan totient function $J_2(d)$ is given by:
\begin{eqnarray}
J_2(d)=d^2\prod _{i=1}^N\left (1-\frac{1}{p_i^2}\right )=\psi(d)\varphi (d)
\end{eqnarray} 
Here $\varphi (d)$ is the Euler totient function.
These functions have been used in the context of finite systems in \cite{VB,P}.

A line through the origin in ${\mathbb Z}(d)\times {\mathbb Z}(d)$ is the set
\begin{eqnarray}\label{l}
{\cal L}(\nu,\mu)=\{(\nu \alpha,\mu\alpha)\;|\; \alpha \in {\mathbb Z}(d)\} 
\end{eqnarray}
Mathematically this is a cyclic module generated by $(\nu,\mu)$.
We use a more physical language and we refer to it as a line through the origin.

The matrices 
\begin{eqnarray}\label{SY}
g(\kappa, \lambda |\mu , \nu)\equiv 
\left(\begin{array}{cc}
\kappa& \lambda\\
\mu& \nu\\
\end{array}
\right)
;\;\;\;\;\;\;{\rm det} (g)=\kappa \nu-\lambda \mu=1\;({\rm mod}\; d);\;\;\;\;\;
\kappa, \lambda, \mu, \nu \in {\mathbb Z}(d).
\end{eqnarray}
form the group $Sp(2,{\mathbb Z}(d))$.
The cardinality of this group is $dJ_2(d)$.

Acting with the matrix $g(\kappa, \lambda |\mu , \nu)$ 
on a point $(\rho, \sigma)\in {\mathbb Z}(d)\times {\mathbb Z}(d)$ 
we get the point $(\kappa \rho +\lambda \sigma, \mu \rho+\nu \sigma)$.
Acting with the symplectic matrix $g(\kappa, \lambda |\mu , \nu)$
on all the points of a line  ${\cal L}(\rho, \sigma)$ we get all the points of the line 
${\cal L}(\kappa \rho +\lambda \sigma, \mu \rho+\nu \sigma)$ for which we also use the notation
$g(\kappa, \lambda |\mu , \nu){\cal L}(\rho, \sigma)$.

In \cite{SV} we have proved  various properties for lines in
${\mathbb Z}(d)\times {\mathbb Z}(d)$. 
For completeness below we give these properties together with several new properties, but we prove only the new ones.
\begin{definition}
Lines ${\cal L}(\nu,\mu)$ through the origin, with exactly $d$ points are called 
maximal lines.
\end{definition}
\begin{proposition}\label{1234}
\begin{itemize}
\mbox{}

\item[(1)]
If $\lambda$ is an inverible element (unit) in ${\mathbb Z}(d)$
then ${\cal L}(\nu,\mu)={\cal L}(\nu \lambda,\mu \lambda)$, and if 
$\lambda$ is a non-inverible element in ${\mathbb Z}(d)$ then 
${\cal L}(\nu \lambda,\mu \lambda)\subset {\cal L}(\nu,\mu)$.

\item[(2)]
The number of points in ${\cal L}(\nu,\mu)$ is $d/{\mathfrak G}(\nu, \mu, d)$
(where ${\mathfrak G}(\nu, \mu, d)$ is the greatest common divisor of these numbers).
If $d$ is a prime number, all lines are maximal (apart from the ${\cal L}(0,0)$ 
which consists of the origin only).

\item[(3)]
Let $d_i$ be a divisor of $d$.
Lines through the origin with exactly $d_i$ points, consists of points $(\rho, \sigma)$
with coordinates $\rho,\sigma$ in the $(d/d_i){\mathbb Z}(d_i)$ subgroup of ${\mathbb Z}(d)$.

\item[(4)]
Let $d_i$ be a divisor of $d$.
There is a total number of $\psi(d_i)$ lines through the origin with $d_i$ points.
Consequently, there are $\psi(d)$ maximal lines through the origin.

\item[(5)]
The intersection of two lines ${\cal L}(\nu ,\mu )$ and ${\cal L}(\rho, \sigma)$
is a `subline' and the number of its points is a divisor of $d$.
When $d$ is a prime number, two lines can only have one point in common.
\item[(6)]
The lines ${\cal L}(\rho, \sigma)$ and $g(\kappa, \lambda |\mu , \nu){\cal L}(\rho, \sigma)$ 
have the same number of points.

\item[(7)]
In the special case that $d$ is a prime number,
all lines through the origin are given by
\begin{eqnarray}
{\cal L}(0,1);\;\;\;\;g(0,1|-1,-\lambda){\cal L}(0,1);\;\;\;\;\lambda=0,1,..., d-1
\end{eqnarray}
\end{itemize}
\end{proposition}

\begin{proof}
Parts (1),(2) and (5) have been proved in \cite{SV} (proposition 2.1).
\begin{itemize}
\item[(3)]
The number of points in a line ${\cal L}(\rho, \sigma)$
through the origin, is $d/{\mathfrak G}(\rho, \sigma, d)$.
Therefore for lines with $d_i$ points, ${\mathfrak G}(\rho, \sigma, d)=d/d_i$.
This means that $\rho=\rho 'd/d_i$ where $\rho '$ is an element of ${\mathbb Z}(d_i)$ 
(and similarly for $\sigma$).
 
\item[(4)]
We have proved above that the lines ${\cal L}(\nu,\mu)$ with $d_i$ points have points with
components $\nu,\mu$ in the $(d/d_i){\mathbb Z}(d_i)$ subgroup of ${\mathbb Z}(d)$.
In addition to that there are $\psi (d)$ lines 
through the origin ${\cal L}(\rho, \sigma)$ where $\rho,\sigma \in {\mathbb Z}(d)$
with exactly $d$ points. Consequently, there are $\psi (d_i)$ lines through the origin with 
$d_i$ points.

\item[(6)]
Since $det(g)\ne 0$ two different points on a line are mapped to two other points which are different 
from each other. This proves the statement.

\item[(7)]
We first point out that the line ${\cal L}(0, \tau)$ (with $\tau \ne 0$) is the same as the line 
${\cal L}(0, 1)$. 
We next prove that any other line ${\cal L}(\rho, \sigma)$ (which will have $\rho \ne 0$)
can be written as $g(0,1|-1,-\lambda){\cal L}(0,1)$.
We need to prove that given a point $(\alpha \rho, \alpha \sigma)$  on the line
${\cal L}(\rho, \sigma)$ (where $\alpha \ne 0$ for points other than the origin),
there exist a point  $(0, \beta)$ on the line
${\cal L}(0, 1)$ such that $\beta=\alpha \rho$ and $-\lambda \beta =\alpha \sigma$.
But this is true for $\beta =\alpha \rho$ and $\lambda=-\rho ^{-1}\sigma$ (the $\rho ^{-1}$ exists because
$d$ is a prime number and ${\mathbb Z}(d)$ is a field). 
This completes the proof.
\end{itemize}
\end{proof}

The $\psi(d)$ maximal lines through the origin, have $d$ points each.
Except the origin, these lines have $\psi(d) (d-1)$ points in total, some of which have been counted many times
simply because the geometry is not near-linear and two lines might have many points in common.
But except the origin, there are only $d^2-1$ points in ${\mathbb Z}(d)\times {\mathbb Z}(d)$. 
As a measure of the deviation of our geometry from a near-linear geometry, we introduce the redundancy parameter
\begin{eqnarray}\label{red}
{\mathfrak r}=\frac{\psi(d)(d-1)}{d^2-1}-1=\frac{\psi(d)}{d+1}-1.
\end{eqnarray}
When $d$ is a prime number, two different lines have at most one point in common, the geometry is near-linear, 
and ${\mathfrak r}=0$.

Much work in finite geometries is on near-linear geometries where
two points belong to at most one line \cite{f1,f2,f3}.
Our geometry is {\bf not} a near-linear geometry 
because two lines can have $d_1$ points in common, where $d_1|d$ ($d_1$ is a divisor of $d$). 
Only in the case of prime $d$, the geometry is a near-linear geometry.

Different aspects of the ${\mathbb Z}(d)\times {\mathbb Z}(d)$ as a finite geometry have been studied in 
\cite{r1,r2,r3,r4,r5,r6}. In particular we note that the isotropic lines in \cite{r5} correspond to the 
maximal lines in the present work. In the present paper we emphasize the fact that two lines may have 
many points in common, and this leads to the concept of sublines and to the 
`geometric redundancy' in Eq.(\ref{red}).
Later this geometric redundancy, will be related through duality to 
redundancy in the weak mutually unbiased bases.

\subsection{Factorization}

Let $d=d_1...d_N$ where any pair of $d_i,d_j$ are coprime.
Based on the Chinese remainder theorem, we can introduce one-to-one maps between
${\mathbb Z}(d)$ and ${\mathbb Z}(d_1)\times...\times {\mathbb Z}(d_N)$.
Good used them in fast Fourier transforms\cite{Good}.

Below we consider the case $d=d_1d_2$ where $d_1,d_2$ are prime numbers.
Let
\begin{equation}\label{20}
r_1=d_2;\;\;\;\;r_2=d_1;\;\;\;\;
t_i r_i=1\;( mod\ d_i);\;\;\;\;\;
s_i=t_i r_i\in {\mathbb Z}(d).
\end{equation}
The first map is 
\begin{eqnarray}\label{map1}
m\leftrightarrow (m_1,m_2);\;\;\;\;\;
m_i=m (mod\ d_i);\;\;\;\;m=m_1s_1+m_2s_2
\end{eqnarray}
The second map is 
\begin{eqnarray}\label{map2}
m\leftrightarrow (\overline m_1,\overline m_2);\;\;\;\;\;\overline m_i=mt_i=m_it_i(mod\ d_i);\;\;\;\;\;
m=\overline m_1 r_1+ \overline m_2 r_2 (mod\ d)
\end{eqnarray}
We next introduce the following bijective map
\begin{eqnarray}
{\mathbb Z}(d)\times {\mathbb Z}(d)\;\;\leftrightarrow\;\;[{\mathbb Z}(d_1)\times{\mathbb Z}(d_2)]
\times [{\mathbb Z}(d_1)\times {\mathbb Z}(d_2)]
\end{eqnarray}
given by 
\begin{eqnarray}\label{map100}
(m,n)\leftrightarrow (m_1,m_2,\overline n_1,\overline n_2).
\end{eqnarray}
We use here the first map of Eq.(\ref{map1}) for $m \leftrightarrow (m_1,m_2)$
and the second  map of Eq.(\ref{map2}) for $n \leftrightarrow (\overline n_1,\overline n_2)$.
In ${\mathbb Z}(d_i)\times{\mathbb Z}(d_i)$ we consider the lines
\begin{eqnarray}
{\cal L}^{(i)}(\nu _i,\mu _i)=\{(\nu _i\alpha _i,\mu _i\alpha _i)\;|\; 
\alpha _i\in {\mathbb Z}(d_i)\}. 
\end{eqnarray}
We can rewrite the 
line  ${\cal L}(\nu,\mu)$ in ${\mathbb Z}(d)\times{\mathbb Z}(d)$ given in Eq.(\ref{l}), in terms
of the two lines ${\cal L}^{(i)}(\nu _i,\mu_i)$ in ${\mathbb Z}(d_i)\times{\mathbb Z}(d_i)$, as
\begin{eqnarray}\label{589}
{\cal L}(\nu,\mu)=
{\cal L}^{(1)}(\nu _1,\overline \mu _1)\times {\cal L}^{(2)}(\nu _2, \overline \mu _2) 
\end{eqnarray}
To prove this we use the fact that if $\alpha \leftrightarrow (\alpha _1,\alpha _2)$ and 
$\nu\leftrightarrow (\nu_1,\nu _2)$
then $\alpha \nu\leftrightarrow (\alpha _1\nu_1,\alpha _2\nu_2)$.
In the case $(\nu _1,\overline \mu _1)\ne(0,0)$ and $(\nu _2,\overline \mu _2)\ne(0,0)$
the lines ${\cal L}^{(1)}(\nu _1,\overline \mu _1)$ and ${\cal L}^{(2)}(\nu _2, \overline \mu _2)$
have $d_1$ and $d_2$ points correspondingly, and the line
${\cal L}(\nu,\mu)$ has $d_1d_2$ points.
We refer to ${\cal L}^{(1)}(\nu _1,\overline \mu _1)$ and ${\cal L}^{(2)}(\nu _2, \overline \mu _2)$ 
as the first component line and the second component line of ${\cal L}(\nu,\mu)$.

\begin{proposition}
\mbox{}
\begin{itemize}
\item[(1)]
The set of $\psi(d)$ maximal lines in ${\mathbb Z}(d)\times {\mathbb Z}(d)$ through the origin,
(where $d=d_1d_2$ and $d_1,d_2$ are prime numbers)
is given by
\begin{eqnarray}\label{AS}
&&{\cal L}_1={\cal L}^{(1)}(0,1)\times {\cal L}^{(2)}(0,1)
;\;\;\;\;\nonumber\\
&&{\cal L}_{2+\lambda _2}=
{\cal L}^{(1)}(0,1)\times [g(0,1|-1,-\lambda _2){\cal L}^{(2)}(0,1)]\nonumber\\
&&{\cal L}_{2+d_2+\lambda _1}=
[g(0,1|-1,-\lambda _1){\cal L}^{(1)}(0,1)]\times {\cal L}^{(2)}(0,1)\nonumber\\
&&{\cal L}_{2+d_1+d_2+\lambda _2+\lambda _1d_2}=
[g(0,1|-1,-\lambda _1){\cal L}^{(1)}(0,1)]\times [g(0,1|-1,-\lambda _2){\cal L}^{(2)}(0,1)]
\end{eqnarray}
where $0\le \lambda _1\le d_1-1$ and $0\le \lambda _2\le d_2-1$.
\item[(2)]
These maximal lines through the origin, can also be derived through symplectic transformations in ${\mathbb Z}(d)\times {\mathbb Z}(d)$ 
acting on the points of the line ${\cal L}_1$, as follows:
\begin{eqnarray}\label{AS1}
&&{\cal L}_{2+\lambda _2}=g(s_1,t_2s_2|-d_1,s_1-\lambda _2s_2){\cal L}_1\nonumber\\
&&{\cal L}_{2+d_2+\lambda _1}=g(s_2,t_1s_1|-d_2,s_2-\lambda _1s_1){\cal L}_1\nonumber\\
&&{\cal L}_{2+d_1+d_2+\lambda _2+\lambda _1d_2}=
g(0,\eta|-d_1-d_2,-\lambda _1s_1-\lambda_2s_2){\cal L}_1
\end{eqnarray}
Here $\eta=t_1^2d_2+t_2^2d_1$.

\item[(3)]
Acting with the matrices $g(\kappa, \lambda |\mu , \nu)$ on a maximal line ${\cal L}(\rho, \sigma)$
through the origin, we get all $\psi (d)$ maximal lines through the origin.

\end{itemize}
\end{proposition}
\begin{proof}
\mbox{}
\begin{itemize}
\item[(1)]
The lines ${\cal L} ^{(1)}(0,1)$, ${\cal L}^{(2)}(0,1)$ have $d_1, d_2$
points, correspondingly.
According to proposition \ref{1234} the lines 
$g(0,1|-1,-\lambda _1){\cal L}^{(1)}(0,1)$, $g(0,1|-1,-\lambda _2){\cal L}^{(2)}(0,1)$ also have  
$d_1, d_2$ points, correspondingly.
Consequently all lines in Eq.(\ref{AS}) have $d=d_1d_2$ points each, and they are maximal lines.
These lines are different from each other, and Eqs.(\ref{AS}) give
$\psi(d)=(d_1+1)(d_2+1)$ such lines.
Since the number of maximal lines in ${\mathbb Z}(d)\times {\mathbb Z}(d)$ through the origin
is $\psi(d)$, Eqs.(\ref{AS}) give all maximal lines.
\item[(2)]
We have proved in \cite{SV} that 
\begin{eqnarray}\label{D175}
g(\kappa, \lambda |\mu, \nu)=g^{(1)}(\kappa _1, \lambda _1 r_1|\overline \mu _1, \nu _1)\otimes
g^{(2)}(\kappa _2, \lambda _2 r_2|\overline \mu _2, \nu _2)
;\;\;\;\;
\kappa, \lambda, \mu, \nu \in {\mathbb Z}(d)
\end{eqnarray}
Here the $\kappa _i, \lambda _i , \nu _i \in {\mathbb Z}(d_i)$ are the components of 
$\kappa , \lambda , \nu $ correspondingly, according to the map of Eq.(\ref{map1}).
${\overline \mu} _i \in {\mathbb Z}(d_i)$ are the components of $\mu$ 
according to the map of Eq.(\ref{map2}).
Using this on the lines in Eq.(\ref{AS}) which are products of two lines in
${\mathbb Z}(d_1)\times {\mathbb Z}(d_1)$ and ${\mathbb Z}(d_2)\times {\mathbb Z}(d_2)$,
we get the lines in Eq.(\ref{AS1}) in ${\mathbb Z}(d)\times {\mathbb Z}(d)$.

\item[(3)]

For $d=d_1d_2$ where $d_1,d_2$ are prime numbers, given an arbitrary line factorized as in Eq.(\ref{589}),
we act with symplectic transformations and using Eq.(\ref{D175}) we get
\begin{eqnarray}\label{5890}
g(\kappa, \lambda |\mu, \nu){\cal L}(\nu,\mu)=g^{(1)}(\kappa _1, \lambda _1 r_1|\overline \mu _1, \nu _1)
{\cal L}^{(1)}(\nu _1,\overline \mu _1)\times 
g^{(2)}(\kappa _2, \lambda _2 r_2|\overline \mu _2, \nu _2){\cal L}^{(2)}(\nu _2, \overline \mu _2) 
\end{eqnarray}
For a prime number, Eq.(6) shows that we can act with symplectic transformations on
${\cal L}(0,1)$ and get all lines through the origin (these lines are maximal lines).
Therefore we can get any line by acting with symplectic transformations on any other line.
This in conjuction with Eq.(\ref{5890}) proves the statement.

\end{itemize}
\end{proof}

We note that $d_i$ are prime numbers and the ${\mathbb Z}(d_i)\times{\mathbb Z}(d_i)$ 
geometry is near-linear geometry, i.e., two lines have at most one point in common.
The following proposition explains how in ${\mathbb Z}(d)\times{\mathbb Z}(d)$
two lines may have one or $d_1$ or $d_2$ points in common.

\begin{proposition}
Let 
${\cal L}(\nu,\mu)={\cal L} ^{(1)}(\nu _1,\overline \mu _1)\times {\cal L}^{(2)}(\nu _2, \overline \mu _2)$ 
and
${\cal L}(\nu',\mu')={\cal L}^{(1)}(\nu _1',\overline \mu _1')\times {\cal L}^{(2)}(\nu _2',\overline \mu _2')$
be a pair of maximal lines through the origin in ${\mathbb Z}(d)\times{\mathbb Z}(d)$.
There are $\psi(d)[\psi (d)-1]/2$ such pairs
(here the pairs are not ordered and each pair contains two lines which are different from each other).
Then one of the following holds:
\begin{itemize}
\item[(1)]
The two lines have $d_2$ points in common if and only if the second component line is the same
(i.e., $\nu_2=\alpha \nu_2'$ and $\overline \mu_2=\alpha \overline \mu_2'$ where 
$\alpha\in {\mathbb Z}(d_2)$ and $\alpha \ne 0$).
There are $d_1\psi(d)/2$ such pairs of lines.

\item[(2)]
The two lines have $d_1$ points in common if and only if the first component line is the same
(i.e., $\nu_1=\alpha \nu_1'$ and $\overline \mu_1=\alpha \overline \mu_1'$ where 
$\alpha\in {\mathbb Z}(d_1)$ and $\alpha \ne 0$).
There are $d_2\psi(d)/2$ such pairs of lines.

\item[(3)]
The two lines have only the origin in common, if and only if both component lines are different.
There are $d\psi(d)/2$ such pairs of lines.
\end{itemize}

\end{proposition}
\begin{proof}
We have $\psi(d)$ maximal lines and from this follows that there are 
$\psi(d)[\psi (d)-1]$ ordered pairs of lines (where the lines in each pair are different from each other).
Therefore we have $\psi(d)[\psi (d)-1]/2$ pairs of lines which are not ordered.
 
We next consider three cases:
\begin{itemize}
\item[(1)]
The second component lines of the two lines are the same, and the first component lines
${\cal L}^{(1)}(\nu _1',\overline \mu _1')$ and ${\cal L}^{(1)}(\nu _1,\overline \mu _1)$
are maximal lines in ${\mathbb Z}(d_1)\times{\mathbb Z}(d_1)$
and they have only the origin $(0,0)$ in common (because $d_1$ is a prime number).
The $(0,0)$ combined with the $d_2$ points in ${\cal L}^{(2)}(\nu _2, \overline \mu _2)$
as described in Eq.(\ref{map100}), produce the $d_2$ common points between the lines
${\cal L} ^{(1)}(\nu _1,\overline \mu _1)\times {\cal L}^{(2)}(\nu _2, \overline \mu _2)$
and ${\cal L}^{(1)}(\nu _1',\overline \mu _1')\times {\cal L}^{(2)}(\nu _2,\overline \mu _2)$.

It is easily seen that there are 
$[d_1(d_1+1)][d_2+1]/2$ pairs of lines, and this is equal to $d_1\psi (d)/2$.

\item[(2)]
The proof here is analogous to the previous case.

\item[(3)]
Both component lines in the first line, are different from their counterparts in the second line.
If the two lines had a point $(\nu,\mu)\ne (0,0)$ in common, then
the corresponding point $(\nu _1,\overline \mu _1)$ would be a common point of the first component lines
and the corresponding point $(\nu _2,\overline \mu _2)$ would be a common point of the second component lines.
But at least one of the $(\nu _1,\overline \mu _1)$
and $(\nu _2,\overline \mu _2)$ would be different from
$(0,0)$. This would mean that lines in ${\mathbb Z}(d_i)\times{\mathbb Z}(d_i)$ would have more 
than one point in common, which is not possible because $d_i$ is a prime number.

There are $[d_1(d_1+1)][d_2(d_2+1)]/2$ pairs of lines in this case, and this is equal to $d\psi (d)/2$.
  
\end{itemize}

The inverse of these statements, follows immediately from the fact that above we have 
exhausted all possible cases
where the two lines have $1,d_1,d_2$ points in common.

\end{proof}

\begin{remark}\label{Rema}
We consider the set ${\cal S}$ of all maximal lines through the origin in
${\mathbb Z}(d)\times {\mathbb Z}(d)$ (where $d=d_1d_2$ and $d_1,d_2$ are prime numbers with
$d_1<d_2$).
We will introduce a partition of this set such that any two lines in each subset intersect only at the origin.
In order to do this we introduce the notation
\begin{eqnarray}
{\mathfrak L}^{(1)}_{-1}&=&{\cal L}(0,1)\nonumber\\
{\mathfrak L}^{(1)}_{\lambda}&=&g(0,1|-1,-\lambda){\cal L}(0,1);\;\;\;\;\;\lambda=0,...,d_1-1
\end{eqnarray} 
The lines ${\mathfrak L}^{(1)}_i$ are defined for $i=-1,...,d_1-1$.
It will be convinient to take $i\in {\mathbb Z}(d_1+1)$.
Using similar notation for the second subsystem, we consider the sets
\begin{eqnarray}
{\cal S}_n=\{{\mathfrak L}^{(1)}_i\times {\mathfrak L}^{(2)}_{i+n}\;|\;i\in {\mathbb Z}(d_1+1)\}
;\;\;\;\;\;n\in {\mathbb Z}(d_2+1)
\end{eqnarray}
The cardinality of ${\cal S}_n$ is $d_1+1$ and the cardinality of ${\cal S}$ is
$\psi(d)=(d_1+1)(d_2+1)$. 
Any two lines in the same subset ${\cal S}_n$ intersect only at the origin.
Lines in different subsets might intersect at more points.
\end{remark}

\section{Quantum systems with variables in ${\mathbb Z}(d)$}

We consider a quantum system with positions and momenta in ${\mathbb Z}(d)$ (with odd $d$).
The corresponding Hilbert space $H_d$ is $d$-dimensional.
We consider two bases, the positions $|{X};m\rangle$ and 
the momenta $|{P};m\rangle$ which are related through the Fourier transform:
\begin{equation}\label{PPPT}
|{P};n\rangle=F|{X};n\rangle;\;\;\;\;
F=d^{-1/2}\sum _{m,n}\omega (mn)\ket{X;m}\bra{X;n};\;\;\;\;
\omega(m)=\exp \left (i\frac {2\pi m}{d}\right ).
\end{equation}
We also consider the displacement operators 
\begin{eqnarray}
&&{Z}^ \alpha =\sum _{n\in {\mathbb Z}_d}\omega (n \alpha)|{X};n\rangle \langle {X};n|;\;\;\;\;\;
{X}^\beta=\sum _{n\in {\mathbb Z}_d}\omega (-n \beta)|{P};n\rangle \langle {P};n|\nonumber\\
&&X^\beta \;Z^\alpha = Z^\alpha \;X^\beta \;\omega (-\alpha \beta);\;\;\;\;\;X^d=Z^d={\bf 1}
\end{eqnarray}
General displacement operators are 
${D}(\alpha, \beta)={Z}^\alpha \;{X}^\beta \;\omega (-2^{-1}\alpha \beta)$.

We also define the symplectic transformations as
\begin{eqnarray}\label{DD175}
&&X'=S(\kappa, \lambda |\mu, \nu)\;X\;[S(\kappa, \lambda |\mu, \nu)]^{\dagger}
=D(\lambda ,\kappa)\nonumber\\
&&Z'=S(\kappa, \lambda |\mu, \nu)\;Z\;[S(\kappa, \lambda |\mu, \nu)]^{\dagger}=D(\nu,\mu)\nonumber\\
&&\kappa \nu -\lambda \mu=1;\;\;\;\;\;\kappa, \lambda, \mu, \nu \in {\mathbb Z}_d.
\end{eqnarray}
Explicit formulas for $S(\kappa, \lambda |\mu, \nu)$ have been given in \cite{1,Vie}.
Acting with them on the positions and momentum bases we get new bases:
\begin{eqnarray}\label{Sym7}
\ket{X(\kappa, \lambda|\mu,\nu);n} \equiv S(\kappa, \lambda|\mu,\nu)\ket{X;n};\;\;\;\;\;
\ket{P(\kappa, \lambda|\mu,\nu);n} \equiv S(\kappa, \lambda|\mu,\nu)\ket{P;n}
\end{eqnarray}

Mutually unbiased bases are a set
of orthonormal bases $\ket{{\mathfrak X}_i;n}$ where $n\in {\mathbb Z}(d)$ 
and the index $i$ takes values in some set $S$, such that 
for all $i,j\in S$ with $i\ne j$
\begin{eqnarray}\label{8}
|\langle {\mathfrak X}_i;n\ket{{\mathfrak X}_j;m}|=d^{-1/2}
\end{eqnarray}
It is known that for a prime number $d$ the states 
\begin{eqnarray}\label{S100}
\ket{X;m};\;\;\;\;\ket{X(0,1|-1,-\lambda );m};\;\;\;\;\;\lambda=0,1,...,d-1
\end{eqnarray}
are mutually unbiased.

Based on the maps  in Eqs(\ref{map1}), (\ref{map2})
we can factorize a system with variables in ${\mathbb Z}(d)$
in terms of two subsystems with variables in ${\mathbb Z}(d_1)$ and ${\mathbb Z}(d_2)$
(where $d=d_1d_2$ and $d_1,d_2$ are prime numbers)\cite{FA1,FA2}. 
There is an isomorphism between the Hilbert space $H_d$ 
and the product of the Hilbert spaces of the subsystems 
${H}_{d_1}\otimes {H}_{d_2}$, where
the position states in $H_d$ correspond to the products of position states in $H_{d_i}$, as follows: 
\begin{equation}\label{nn}
|X; m\rangle \leftrightarrow |X^{(1)}; \overline m_1\rangle \otimes|X^{(2)}; \overline m_2\rangle.
\end{equation}
The momentum states in $H_d$ correspond to the products of 
position states in $H_{d_i}$, as follows: 
\begin{equation}\label{nn}
|P;m\rangle \leftrightarrow |P^{(1)}; m_1\rangle \otimes |P^{(2)}; m_2\rangle
\end{equation}

\section{Weak mutually unbiased bases}

Below $d=d_1d_2$ where $d_1,d_2$ are prime numbers (different than $2$). 

\begin{definition}\label{def}
We consider a set $\ket{{\mathfrak B}_i;m}$ of $\ell$ orthonormal bases in the Hilbert space $H_d$,
where $m\in{\mathbb Z}(d)$ and $i=1,...,\ell$.
Let
\begin{eqnarray}
f_{ji}(n,m)=|\langle {\mathfrak B}_j;n \ket{{\mathfrak B}_i;m}|;\;\;\;\;\;f_{ij}(m,n)=f_{ji}(n,m)
\end{eqnarray}
We call them weak mutually unbiased bases if for any pair of different bases
($i\ne j$), one of the following holds:
\begin{itemize}
\item[(1)]
\begin{eqnarray}\label{B1}
f_{ji}(n,m)&=&d_1^{-1/2};\;\;{\rm for\;the}\;d_1d\;{\rm pairs}\; (n,m)\in {\mathbb Z}(d)\times {\mathbb Z}(d)\;\; {\rm such\;\; that}\;\;n=m\;({\rm mod}\;d_2)\;\;
\nonumber\\
f_{ji}(n,m)&=&0;\;\;\;{\rm for \;\;the \;\;rest}\;\;(n,m)\;\;  
\end{eqnarray}
 
\item[(2)]
\begin{eqnarray}\label{B2}
f_{ji}(n,m)&=&d_2^{-1/2};\;\;{\rm for\;the}\;d_2d\;{\rm pairs}\;(n,m)\in {\mathbb Z}(d)\times {\mathbb Z}(d)\;\; {\rm such\;\; that}\;\;n=m\;({\rm mod}\;d_1)
\nonumber\\
f_{ji}(n,m)&=&0;\;\;\;{\rm for \;\;the \;\;rest}\;\;(n,m)
\end{eqnarray}
 
\item[(3)]
\begin{eqnarray}\label{B3}
f_{ji}(n,m)=(d_1d_2)^{-1/2};\;\;\;\;{\rm for\;\;all}\;\;(n,m)
\in {\mathbb Z}(d)\times {\mathbb Z}(d)
\end{eqnarray}

\end{itemize}
\end{definition}

\begin{theorem}
\begin{itemize}
\mbox{}
\item[(1)]
Any set of weak mutually unbiased bases in $H_d$ can be written as
$\ket{B_j^{(1)};{\overline m_1}}\otimes \ket{B_j^{(2)};{\overline m_2}}$
where $\ket{B_j^{(1)};{\overline m_1}}$ is a set of mutually unbiased bases in
$H_{d_1}$, and $\ket{B_j^{(2)};{\overline m_2}}$ is a set of mutually unbiased bases in
$H_{d_2}$.
Some of the  bases $\ket{B_i^{(1)};{\overline m_1}}$ with different index $i$ might be the same basis,
and similarly for $\ket{B_i^{(2)};{\overline m_2}}$.

\item[(2)]
The maximum number of weak mutually unbiased bases is $\psi(d)$.
In this case, there are $\psi(d)[\psi(d)-1]/2$ sets of values $f_{ji}(n,m)$ 
with $i,j=1,..., \psi (d)$ and $i\ne j$.
From them, $d_1\psi(d)/2$ belong to the first category of Eq.(\ref{B1}),
$d_2\psi(d)/2$ belong to the second category of Eq.(\ref{B2}),
and $d\psi(d)/2$ belong to the third category of Eq.(\ref{B3}).

\end{itemize}
\end{theorem}

\begin{proof}
\mbox{}
\begin{itemize}

\item[(1)]
We are given a set $\ket{{\mathfrak B}_i;m}$ of $\ell$ orthonormal bases in the $d$-dimensional 
Hilbert space $H_d$,
where $m\in{\mathbb Z}(d)$ and $i=1,...,\ell$.
As explained above, using Eq.(\ref{map2}) this can be factorized as 
\begin{eqnarray}
\ket{{\mathfrak B}_i;m}=\ket{B_i^{(1)};{\overline m_1}}\otimes \ket{B_i^{(2)};{\overline m_2}}
;\;\;\;\;\;i=1,...,\ell
\end{eqnarray}
where $\ket{B_i^{(1)};{\overline m_1}}$ is an orthonormal basis in $H_{d_1}$ and
$\ket{B_i^{(2)};{\overline m_2}}$ is an orthonormal basis in $H_{d_2}$.
We need to prove that the set $\{\ket{B_i^{(1)};{\overline m_1}}|i=1,...,\ell\}$, 
is a set of mutually unbiased bases in $H_{d_1}$.
The same applies to the second subsystem.

We consider the following three cases:
\begin{itemize}
\item[(a)]
For $i\ne j$ we consider the case of Eq.(\ref{B1}). 
According to the map of Eq.(\ref{map2}), $n\leftrightarrow (\overline n_1,\overline n_2)$
and $m\leftrightarrow (\overline m_1,\overline m_2)$ and for $n=m\;({\rm mod}\;d_2)$ we get
\begin{eqnarray}\label{qa1}
|\langle B_i^{(1)};{\overline n_1}\ket{B_j^{(1)};{\overline m_1}}|=\frac{1}{d_1^{1/2} 
|\langle B_i^{(2)};{\overline n_2}\ket{B_j^{(2)};{\overline m_2}}|}
\end{eqnarray}

The condition $n=m\;({\rm mod}\;d_2)$ gives $\overline n_2=\overline m_2$.
As $(n,m)$ take all values in ${\mathbb Z}(d)\times {\mathbb Z}(d)$ such that $n=m\;({\rm mod}\;d_2)$, 
the $(\overline n_1,\overline m_1)$ take all values in ${\mathbb Z}(d_1)\times {\mathbb Z}(d_1)$.

From Eq.(\ref{qa1}) follows that
$d_1^{-1/2}\le |\langle B_i^{(1)};{\overline n_1}\ket{B_j^{(1)};{\overline m_1}}|\le 1$.
In addition to that 
\begin{eqnarray}
\sum _{{\overline m_1}\in {\mathbb Z}(d_1)}|\langle B_i^{(1)};{\overline n_1}\ket{B_j^{(1)};{\overline m_1}}|^2=1
\end{eqnarray}
From this we conclude that $|\langle B_i^{(1)};{\overline n_1}\ket{B_j^{(1)};{\overline m_1}}|=d_1^{-1/2}$, for
all $(\overline n_1,\overline m_1)$ in ${\mathbb Z}(d_1)\times {\mathbb Z}(d_1)$.
Therefore the $\ket{B_i^{(1)};{\overline n_1}}$ are mutually unbiased bases in $H_{d_1}$.
Since there are $d_1+1$ mutually unbiased bases in $H_{d_1}$, 
some of the bases $\ket{B_i^{(1)};{\overline n_1}}$ with different index $i$, might be the same basis.

In this case $|\langle B_i^{(2)};{\overline n_2}\ket{B_j^{(2)};{\overline m_2}}|=1$
and therefore
$\ket {B_i^{(2)};{\overline m_2}}$ and $\ket{B_j^{(2)};{\overline n_2}}$ 
are the same basis in $H_{d_2}$.
Here ${\overline n_2}={\overline m_2}$ because $n=m\;({\rm mod}\;d_2)$.

\item[(b)]
The second case corresponding to Eq.(\ref{B2}) is similar to the above case.
\item[(c)]
For $i\ne j$ we consider the case of Eq.(\ref{B3}). 
According to the map of Eq.(\ref{map2}), $n\leftrightarrow (\overline n_1,\overline n_2)$
and $m\leftrightarrow (\overline m_1,\overline m_2)$, and we get
\begin{eqnarray}\label{qa2}
|\langle B_i^{(1)};{\overline n_1}\ket{B_j^{(1)};{\overline m_1}}|
|\langle B_i^{(2)};{\overline n_2}\ket{B_j^{(2)};{\overline m_2}}|=(d_1d_2)^{-1/2}
\end{eqnarray}
We will prove that $|\langle B_i^{(2)};{\overline n_2}\ket{B_j^{(2)};{\overline m_2}}|$
is constant for all $\overline m_2 \in {\mathbb Z}(d_2)$.
We consider the overlap of the vector  $\ket{{\mathfrak B}_j;n}$ in the $i$-basis as in Eq.(\ref{qa2}),
with another vector $\ket{{\mathfrak B}_j;m'}$ in the $j$-basis such that $m=m'\;({\rm mod}\; d_1)$.
Then $\overline m_1=\overline m_1'$ and as
$m'$ takes all values in ${\mathbb Z}(d)$ such that $m=m'\;({\rm mod}\; d_1)$,
the $\overline m_2'$ takes all values in ${\mathbb Z}(d_2)$. 
We get
\begin{eqnarray}\label{qa3}
|\langle B_i^{(1)};{\overline n_1}\ket{B_j^{(1)};{\overline m_1}}|
|\langle B_i^{(2)};{\overline n_2}\ket{B_j^{(2)};{\overline m_2}'}|=(d_1d_2)^{-1/2}
\end{eqnarray}
From Eqs.(\ref{qa2}),(\ref{qa3}) we see that $|\langle B_i^{(2)};{\overline n_2}\ket{B_j^{(2)};{\overline m_2}}|$
is constant for all $\overline m_2 \in {\mathbb Z}(d_2)$. 
This in conjuction with the relation
\begin{eqnarray}\label{qa4}
\sum _{{\overline m_2}\in {\mathbb Z}(d_2)}|\langle B_i^{(2)};{\overline n_2}\ket{B_j^{(2)};{\overline m_2}}|^2=1
\end{eqnarray}
leads to the result that $|\langle B_i^{(2)};{\overline n_2}\ket{B_j^{(2)};{\overline m_2}}|=d_2^{-1/2}$.
In a similar way we prove that  
$|\langle B_i^{(1)};{\overline n_1}\ket{B_j^{(1)};{\overline m_1}}|=d_1^{-1/2}$.
Therefore the $\ket{B_i^{(1)};{\overline m_1}}$ are mutually unbiased bases in $H_{d_1}$ and
$\ket{B_i^{(2)};{\overline m_2}}$ are mutually unbiased bases in $H_{d_2}$.

\end{itemize}
\item[(2)]
Since the number of mutually unbiased bases in $H_{d_1}$ is $d_1+1$ and in $H_{d_2}$ is $d_2+1$, we conclude that
the maximum number of weak mutually unbiased bases is $\psi (d)=(d_1+1)(d_2+1)$. 
It is now clear that there are $\psi(d)[\psi(d)-1]/2$ sets of values $f_{ji}(n,m)$ 
with $i,j=1,..., \psi (d)$ and $i\ne j$ and we now
prove that $d_1\psi(d)/2$ of them belong to the first category of Eq.(\ref{B1}).

We proved earlier that in the case of Eq(\ref{B1}), 
\begin{eqnarray}
\ket{{\mathfrak B}_i;m}=\ket{B_i^{(1)};{\overline m_1}}\otimes \ket{B_i^{(2)};{\overline m_2}}
;\;\;\;\;\;\ket{{\mathfrak B}_j;m}=\ket{B_j^{(1)};{\overline n_1}}\otimes \ket{B_j^{(2)};{\overline n_2}}
\end{eqnarray}
where $\ket{B_i^{(1)};{\overline m_1}}$ and $\ket{B_j^{(1)};{\overline n_1}}$
are mutually unbiased bases in $H_{d_1}$ and 
$\ket{B_i^{(2)};{\overline m_2}}$ and $\ket{B_j^{(2)};{\overline n_2}}$ are the same basis in $H_{d_2}$.
We have $(d_1+1)^2-(d_1+1)=d_1(d_1+1)$ pairs of 
$\ket{B_i^{(1)};{\overline m_1}}$,$\ket{B_j^{(1)};{\overline n_1}}$
bases, with $i\ne j$. We multiply this with the number $d_2+1$ of mutually unbiased bases in $H_{d_2}$
and we prove that there are $d_1(d_1+1)(d_2+1)/2=d_1\psi(d)/2$ sets of values $f_{ji}(n,m)$
which belong to the first category of Eq.(\ref{B1}).
The denominator $2$ is due to the symmetry $f_{ij}(m,n)=f_{ji}(n,m)$.

In a similar way we prove that
$d_2\psi(d)/2$ sets of values $f_{ji}(n,m)$ belong to the second category of Eq.(\ref{B2}).
The total number of sets of values $f_{ji}(n,m)$ is $\psi(d)[\psi(d)-1]/2$ and therefore 
$d\psi(d)/2$ of them belong to the third category of Eq.(\ref{B3}).

\end{itemize}
\end{proof}

\begin{remark}
A pair of bases which satisfy Eq.(\ref{B3}) are mutually unbiased.
But bases which belong to different such pairs 
are not necessarily mutually unbiased. 
\end{remark}

\subsection{An explicit construction of weak mutually unbiased bases}

There exists at least one set of weak mutually unbiased bases.
In \cite{SV}, we have introduced a set of $\psi (d)$ orthonormal bases in $H_d$ which are tensor products of 
$d_1+1$ mutually unbiased bases in $H_{d_1}$  with $d_2+1$ mutually unbiased bases in $H_{d_2}$. 
We have shown that this set of bases obeys the requirements in definition 4.1.
Here we summarize briefly this construction in order to present 
in the next section, the duality between 
lines in ${\mathbb Z}(d)\times {\mathbb Z}(d)$ and weak mutually unbiased bases.
\begin{eqnarray}\label{S}
&&\ket{{\mathfrak B}_1;m}=\ket{X^{(1)};{\overline m}_1}\otimes \ket{X^{(2)};{\overline m}_2}=\ket{X(1,0|0,1);{m}}\nonumber\\
&&\ket{{\mathfrak B}_{2+\lambda _2};m}=\ket{X^{(1)};{\overline m}_1}\otimes \ket{X^{(2)}(0,1|-1,-\lambda _2);{\overline m}_2}
=\ket{X(s_1, t_2s_2|-d_1, s_1-\lambda _2s_2);{m}}\nonumber\\
&&\ket{{\mathfrak B}_{2+d_2+\lambda _1};m}=
\ket{X^{(1)}(0,1|-1,-\lambda _1);{\overline m}_1}\otimes \ket{X^{(2)};{\overline m}_2}=
\ket{X(s_2,t_1s_1|-d_2,-\lambda _1s_1+s_2);{m}}\nonumber\\
&&\ket{{\mathfrak B}_{2+d_1+d_2+\lambda _2+\lambda _1d_2};m}=
\ket{X^{(1)}(0,1|-1,-\lambda _1);{\overline m}_1}\otimes \ket{X^{(2)}(0,1|-1,-\lambda _2);{\overline m}_2}
\nonumber\\&&=
\ket{X(0,\eta, |-d_1-d_2, -\lambda_1s_1-\lambda _2 s_2);{m}}
\end{eqnarray}
Here $0\le \lambda _1\le d_1-1$, $0\le \lambda _2\le d_2-1$, and the
variables $s_1,t_1,s_2,t_2$ have been defined in Eq.(\ref{20}).

\begin{remark}
In \cite{Kantor} it has been shown that there many sets of unitarily inequivalent mutually unbiased bases.
This will lead to many unitarily inequivalent weak mutually unbiased bases. 
Below we show the duality between lines in ${\mathbb Z}(d)\times {\mathbb Z}(d)$ and 
the weak mutually unbiased bases constructed explicitly through symplectic transformations, above.
\end{remark}

\begin{remark}\label{Rema1}
This remark is analogous to remark \ref{Rema} for lines.
We use the notation
\begin{eqnarray}\label{AS1}
\ket{{\mathfrak X}_{-1}^{(1)};{\overline m_1}}&=&\ket{X^{(1)};\overline m_1}\nonumber\\
\ket{{\mathfrak X}_{\lambda _1}^{(1)};{\overline m_1}}&=&\ket{X^{(1)}(0,1|-1,-\lambda _1);
{\overline m}_1};\;\;\;\;\lambda _1=0,...,p_1-1,
\end{eqnarray}
and define a partition of the set of weak mutually unbiased bases , as follows:
\begin{eqnarray}
{\cal T}_\ell&=&\{\ket{{\mathfrak X}_{i_1}^{(1)};{\overline m}_1}\otimes \ket{{\mathfrak X}^{(2)}_{i_1+\ell};
{\overline m}_2}\;|\;
i_1\in {\mathbb Z}_{p_1+1}\} ;\;\;\;\ell\in {\mathbb Z}_{p_2+1}
\end{eqnarray}
We have explained in \cite{SV} that
the bases in the same set ${\cal T}_\ell$ are mutually unbiased.
\end{remark}

\section{Duality between lines in ${\mathbb Z}(d)\times {\mathbb Z}(d)$ and weak mutually unbiased bases}

There exists a correspondence (duality) between the lines in ${\mathbb Z}(d)\times {\mathbb Z}(d)$ and the 
weak mutually unbiased bases.
The `dictionary' for this duality is as follows:
\begin{itemize}
\item
The line ${\cal L}_i$ 
corresponds to the basis $\ket{{\mathfrak B}_i;m}$.
We note that the same parameters are used in the symplectic transformations in Eq.(\ref{AS1}) 
for ${\cal L}_i$, and in the
symplectic transformations in Eq.(\ref{S}) for $\ket{{\mathfrak B}_i;m}$.
\item
The $\psi (d)$ maximal lines through the origin, correspond to the $\psi(d)$ weak mutually unbiased bases.
\item
A pair of maximal lines through the origin which have only the origin as common point, 
correspond to a pair of weak mutually unbiased bases with absolute value of the overlap equal to $d^{-1/2}$ (Eq.(\ref{B3})). In this case the pair of bases is mutually unbiased.
We have seen that there are $d\psi(d)/2$ pairs of maximal lines through the origin 
with this property, and also
$d\psi(d)/2$ pairs of weak mutually unbiased bases with the corresponding property.
\item
A pair of maximal lines through the origin which have $d_1$ points in common, 
correspond to a pair of weak mutually unbiased bases with absolute value of the overlap equal to $d_2^{-1/2}$ (Eq.(\ref{B2}))
We have seen that there are $d_2\psi(d)/2$ pairs of maximal lines through the origin 
with this property, and also
$d_2\psi(d)/2$ pairs of weak mutually unbiased bases with the corresponding property.
\item
A pair of maximal lines through the origin which have $d_2$ points in common, 
correspond to a pair of weak mutually unbiased bases with absolute value of the overlap equal to $d_1^{-1/2}$ 
(Eq.(\ref{B1}))
We have seen that there are $d_1\psi(d)/2$ pairs of maximal lines through the origin 
with this property, and also
$d_1\psi(d)/2$ pairs of weak mutually unbiased bases with the corresponding property.
\item
In Eq.(\ref{red}) we introduced a redundancy parameter that measures the deviation of our finite geometry from the near-linear geometries. 
In Eq.(72) of ref.\cite{SV} we introduced another redundancy parameter for the bases that 
we now call weak mutually unbiased bases, and we
have explained that it measures the `tomographical overcompleteness' of these bases.
These two parameters are equal, which indicates that the concept of weak mutually unbiased bases is
taylored for the geometry of ${\mathbb Z}(d)\times {\mathbb Z}(d)$.
In the case of fields (prime $d$), the geometric redundancy is ${\mathfrak r}=0$ and the weak
mutually unbiased bases, become mutually unbiased bases. 

The concept of mutually unbiased bases, is intimately linked to the number of degrees of freedom in 
the density matrix, without any redundancy.
In tomography experiments, probabilities are measured along lines in ${\mathbb Z}(d)\times {\mathbb Z}(d)$.
The fact that lines in this geometry may intersect at many points, leads to the geometric 
redundancy of Eq.(\ref{red}). This makes necessary a redundancy in the bases linked with such measurements, 
and this leads to the concept of weak mutually unbiased bases. 
\end{itemize}

\subsection{Example for the case $d=15$}

We consider the case $d=15$.
In table I we show the $\psi (15)$ maximal lines ${\cal L}(\rho,\sigma)$ through the origin in ${\mathbb Z}(15)\times {\mathbb Z}(15)$ factorized in terms of
the lines ${{\cal L}^{(1)}}(\rho_1,\overline \sigma _1)$ and ${{\cal L}^{(2)}}(\rho_2,\overline \sigma _2)$, in
${\mathbb Z}(3)\times {\mathbb Z}(3)$ and ${\mathbb Z}(5)\times {\mathbb Z}(5)$, correspondingly.
In table II we show the weak mutually unbiased bases $\ket{B_j;m}$ in $H_{15}$ 
factorized in terms of the mutually unbiased bases 
$\ket{B_j^{(1)};\overline m_1}$ and $\ket{B_j^{(2)};\overline m_2}$ in $H_3$ and $H_5$, correspondingly.

As an example we compare the line ${\cal L}_4={\cal L}(3,7)=g(10,12|13,13){\cal L}_1$ with the basis
$\ket{B_4;m}=\ket{X(10,12|12,13);m}$. 
It is seen that symplectic transformations with the same parameters appear in both cases.

According to the factorization in Eq.(\ref{589}), the line ${\cal L}_4={\cal L}(3,7)$ is factorized as 
${\cal L}_4={{\cal L}^{(1)}}(0,2)\times {{\cal L}^{(2)}}(3,4)$. 
We have explained earlier that ${\cal L}(\nu,\mu)={\cal L}(\nu \lambda,\mu \lambda)$
for any invertible element in ${\mathbb Z}(d)$.
Therefore ${{\cal L}^{(1)}}(\rho_1,\overline \sigma _1)=
{{\cal L}^{(1)}}(\lambda _1\rho_1,\lambda _1\overline \sigma _1)$ for any $\lambda _1\in {\mathbb Z}(3)$
with $\lambda _1\ne 0$, and similarly  ${{\cal L}^{(2)}}(\rho_2,\overline \sigma _2)=
{{\cal L}^{(2)}}(\lambda _2\rho_2,\lambda _2\overline \sigma _2)$ for any $\lambda _2\in {\mathbb Z}(5)$
with $\lambda _2\ne 0$. 
From this follows that ${\cal L}_4={{\cal L}^{(1)}}(0,1)\times {{\cal L}^{(2)}}(1,2)$, which appears in the table.
The basis $\ket{B_4;m}$ is factorized as $\ket{B_4;m}=\ket{X^{(1)};\overline m_1}\times \ket{X^{(2)}(0,1|-1,-2);
\overline m_2}$. Details of this factorization have been given in \cite{SV}.

In table III we present a partition of the set of the maximal lines through the origin in 
${\mathbb Z}(15) \times {\mathbb Z}(15)$ (as discussed in remark \ref{Rema}). 
All the lines in the set ${\cal S}_j$ intersect only in the origin.
Lines in different sets may have many points in common.
In table IV we present a partition of the set of the weak mutually unbiased bases (for the case $d=15$). 
Bases in the same column are mutually unbiased bases. (as discussed in remark \ref{Rema1}). 

\section{Discussion}

The concept of mutually unbiased bases, is related to the number of degrees of freedom in 
the density matrix. It appears to be an ideal concept for optimal tomography, because it has no redundancy.
On the other hand we have explained in this paper that in ${\mathbb Z}(d)\times {\mathbb Z}(d)$
phase space, there is a geometric redundancy in the sense that lines may have many points in common.
The geometry is {\bf not} a near-linear geometry and this leads to the concept of sublines.
Probabilities measured along different lines in tomography experiments, are not independent 
(they should obey the constraints in Eq.(34) in ref.\cite{SV}).

The concept of weak mutually unbiased bases incorporates this redundancy.
There is a duality in the sense that the properties of maximal lines have counterparts in the properties of bases.
For simplicity the work has been presented for the special case that $d$ is the product of two prime numbers.
The generalization to a product of many prime numbers is straightforward. For example 
if $d=d_1d_2d_3$ where $d_1, d_2, d_3$ are prime numbers (different from each other), the definition \ref{def} 
will contain all divisors $d_1, d_2, d_3, d_1d_2, d_2d_3, d_1d_3$ (to the power $-1/2$).
In this case the  weak mutually unbiased bases are tensor products of 
mutually unbiased bases in three subsystems with dimensions $d_1,d_2,d_3$.
In the case that $d$ contains powers of prime numbers, we will use Galois fields for the 
labelling of the states in the component subsystems. 

Quantum tomography is an important technique for state reconstruction in finite quantum systems
\cite{W1,W2,W3,W4,W5,W6,W7,W8,W9,W10,W11}.
The weak mutually unbiased bases can be used in this context.
Other techniques in this general context are symmetric informationally complete 
positive operator valued measures \cite{POVM1,POVM2} and designs\cite{d1,d2,d3,d4}. 

In summary we have introduced the concept of weak mutually unbiased bases, which is motivated by the 
geometrical properties of the ${\mathbb Z}(d)\times {\mathbb Z}(d)$ phase space.

\begin{table}[htbp]
\caption{The maximal lines ${\cal L}(\rho,\sigma)$ through the origin in ${\mathbb Z}(15)\times {\mathbb Z}(15)$
and their factorizations in terms of the lines ${{\cal L}^{(1)}}(\rho_1,\overline \sigma _1)$ in ${\mathbb Z}(3)\times {\mathbb Z}(3)$, and ${{\cal L}^{(2)}}(\rho_2,\overline \sigma _2)$ in 
${\mathbb Z}(5)\times {\mathbb Z}(5)$, according to Eq.(\ref{AS}). 
In the calculations, we take into account that ${{\cal L}^{(1)}}(\rho_1,\overline \sigma _1)=
{{\cal L}^{(1)}}(\lambda _1\rho_1,\lambda _1\overline \sigma _1)$ for any $\lambda _1\in {\mathbb Z}(3)$
with $\lambda _1\ne 0$, and similarly  ${{\cal L}^{(2)}}(\rho_2,\overline \sigma _2)=
{{\cal L}^{(2)}}(\lambda _2\rho_2,\lambda _2\overline \sigma _2)$ for any $\lambda _2\in {\mathbb Z}(5)$
with $\lambda _2\ne 0$.}
\centering
\vspace{12pt}
    \begin{tabular}{|r|c|c|c|}
    \hline
       ${\cal L}(\rho,\sigma)$&$g(\kappa,\lambda|\mu,\nu){\cal L}_1$&${{\cal L}^{(1)}}(\rho_1,\overline \sigma _1)$&${{\cal L}^{(2)}}(\rho_2,\overline \sigma _2)$ \\
      \hline       
      ${\cal L}_1={\cal L}(0,1)$&${\cal L}_1$&${{\cal L}^{(1)}}(0,1)$&${{\cal L}^{(2)}}(0,1)$ \\
            \hline
            
      ${\cal L}_2={\cal L}(6,5)$&$g(10,12|12,10){\cal L}_1$&${{\cal L}^{(1)}}(0,1)$&${{\cal L}^{(2)}}(1,0)=g(0,1|-1,0){{\cal L}^{(2)}}(0,1)$ \\
            \hline
            
      ${\cal L}_3={\cal L}(3,1)$&$g(10,12|12,4){\cal L}_1$&${{\cal L}^{(1)}}(0,1)$&${{\cal L}^{(2)}}(1,4)=g(0,1|-1,-1){{\cal L}^{(2)}}(0,1)$ \\
            \hline
            
      ${\cal L}_4={\cal L}(3,7)$&$g(10,12|12,13){\cal L}_1$&${{\cal L}^{(1)}}(0,1)$&${{\cal L}^{(2)}}(1,3)==g(0,1|-1,-2){{\cal L}^{(2)}}(0,1)$ \\
            \hline
            
      ${\cal L}_5={\cal L}(6,11)$&$g(10,12|12,7){\cal L}_1$&${{\cal L}^{(1)}}(0,1)$&${{\cal L}^{(2)}}(1,2)=g(0,1|-1,-3){{\cal L}^{(2)}}(0,1)$ \\
            \hline
            
      ${\cal L}_6={\cal L}(3,4)$&$g(10,12|12,1){\cal L}_1$&${{\cal L}^{(1)}}(0,1)$&${{\cal L}^{(2)}}(1,1)=g(0,1|-1,-4){{\cal L}^{(2)}}(0,1)$ \\
            \hline
            
      ${\cal L}_7={\cal L}(10,3)$&$g(6,5|10,6){\cal L}_1$&${{\cal L}^{(1)}}(1,0)=g(0,1|-1,0){{\cal L}^{(1)}}(0,1)$&${{\cal L}^{(2)}}(0,1)$ \\
            \hline
            
      ${\cal L}_8={\cal L}(10,13)$&$g(6,5|10,11){\cal L}_1$&${{\cal L}^{(1)}}(1,2)=g(0,1|-1,-1){{\cal L}^{(1)}}(0,1)$&${{\cal L}^{(2)}}(0,1)$ \\
            \hline
            
      ${\cal L}_9={\cal L}(5,4)$&$g(6,5|10,1){\cal L}_1$&${{\cal L}^{(1)}}(1,1)=g(0,1|-1,-2){{\cal L}^{(1)}}(0,1)$&${{\cal L}^{(2)}}(0,1)$ \\
            \hline
            
      ${\cal L}_{10}={\cal L}(1,0)$&$g(0,2|7,0){\cal L}_1$&${{\cal L}^{(1)}}(1,0)=g(0,1|-1,0){{\cal L}^{(1)}}(0,1)$&${{\cal L}^{(2)}}(1,0)=g(0,1|-1,0){{\cal L}^{(2)}}(0,1)$ \\
            \hline
            
      ${\cal L}_{11}={\cal L}(1,12)$&$g(0,2|7,9){\cal L}_1$&${{\cal L}^{(1)}}(1,0)=g(0,1|-1,0){{\cal L}^{(1)}}(0,1)$&${{\cal L}^{(2)}}(1,4)=g(0,1|-1,-1){{\cal L}^{(2)}}(0,1)$ \\
            \hline
            
      ${\cal L}_{12}={\cal L}(1,9)$&$g(0,2|7,3){\cal L}_1$&${{\cal L}^{(1)}}(1,0)=g(0,1|-1,0){{\cal L}^{(1)}}(0,1)$&${{\cal L}^{(2)}}(1,3)=g(0,1|-1,-2){{\cal L}^{(2)}}(0,1)$ \\
            \hline
            
      ${\cal L}_{13}={\cal L}(1,6)$&$g(0,2|7,12){\cal L}_1$&${{\cal L}^{(1)}}(1,0)=g(0,1|-1,0){{\cal L}^{(1)}}(0,1)$&${{\cal L}^{(2)}}(1,2)=g(0,1|-1,-3){{\cal L}^{(2)}}(0,1)$ \\
            \hline
            
      ${\cal L}_{14}={\cal L}(1,3)$&$g(0,2|7,6){\cal L}_1$&${{\cal L}^{(1)}}(1,0)=g(0,1|-1,0){{\cal L}^{(1)}}(0,1)$&${{\cal L}^{(2)}}(1,1)=g(0,1|-1,-4){{\cal L}^{(2)}}(0,1)$ \\
            \hline
            
      ${\cal L}_{15}={\cal L}(1,10)$&$g(0,2|7,5){\cal L}_1$&${{\cal L}^{(1)}}(1,2)=g(0,1|-1,-1){{\cal L}^{(1)}}(0,1)$&${{\cal L}^{(2)}}(1,0)=g(0,1|-1,0){{\cal L}^{(2)}}(0,1)$ \\
            \hline
            
      ${\cal L}_{16}={\cal L}(1,7)$&$g(0,2|7,14){\cal L}_1$&${{\cal L}^{(1)}}(1,2)=g(0,1|-1,-1){{\cal L}^{(1)}}(0,1)$&${{\cal L}^{(2)}}(1,4)=g(0,1|-1,-1){{\cal L}^{(2)}}(0,1)$ \\
            \hline
            
      ${\cal L}_{17}={\cal L}(1,4)$&$g(0,2|7,8){\cal L}_1$&${{\cal L}^{(1)}}(1,2)=g(0,1|-1,-1){{\cal L}^{(1)}}(0,1)$&${{\cal L}^{(2)}}(1,3)=g(0,1|-1,-2){{\cal L}^{(2)}}(0,1)$ \\
            \hline
            
      ${\cal L}_{18}={\cal L}(1,1)$&$g(0,2|7,2){\cal L}_1$&${{\cal L}^{(1)}}(1,2)=g(0,1|-1,-1){{\cal L}^{(1)}}(0,1)$&${{\cal L}^{(2)}}(1,2)=g(0,1|-1,-3){{\cal L}^{(2)}}(0,1)$ \\
            \hline
            
      ${\cal L}_{19}={\cal L}(1,13)$&$g(0,2|7,11){\cal L}_1$&${{\cal L}^{(1)}}(1,2)=g(0,1|-1,-1){{\cal L}^{(1)}}(0,1)$&${{\cal L}^{(2)}}(1,1)=g(0,1|-1,-4){{\cal L}^{(2)}}(0,1)$ \\
            \hline
            
      ${\cal L}_{20}={\cal L}(1,5)$&$g(0,2|7,10){\cal L}_1$&${{\cal L}^{(1)}}(1,1)=g(0,1|-1,-2){{\cal L}^{(1)}}(0,1)$&${{\cal L}^{(2)}}(1,0)=g(0,1|-1,0){{\cal L}^{(2)}}(0,1)$ \\
            \hline
            
      ${\cal L}_{21}={\cal L}(1,2)$&$g(0,2|7,4){\cal L}_1$&${{\cal L}^{(1)}}(1,1)=g(0,1|-1,-2){{\cal L}^{(1)}}(0,1)$&${{\cal L}^{(2)}}(1,4)=g(0,1|-1,-1){{\cal L}^{(2)}}(0,1)$ \\
            \hline
            
      ${\cal L}_{22}={\cal L}(1,14)$&$g(0,2|7,13){\cal L}_1$&${{\cal L}^{(1)}}(1,1)=g(0,1|-1,-2){{\cal L}^{(1)}}(0,1)$&${{\cal L}^{(2)}}(1,3)=g(0,1|-1,-2){{\cal L}^{(2)}}(0,1)$ \\
            \hline
            
      ${\cal L}_{23}={\cal L}(1,11)$&$g(0,2|7,7){\cal L}_1$&${{\cal L}^{(1)}}(1,1)=g(0,1|-1,-2){{\cal L}^{(1)}}(0,1)$&${{\cal L}^{(2)}}(1,2)=g(0,1|-1,-3){{\cal L}^{(2)}}(0,1)$ \\
            \hline
            
      ${\cal L}_{24}={\cal L}(1,8)$&$g(0,2|7,1){\cal L}_1$&${{\cal L}^{(1)}}(1,1)=g(0,1|-1,-2){{\cal L}^{(1)}}(0,1)$&${{\cal L}^{(2)}}(1,1)=g(0,1|-1,-4){{\cal L}^{(2)}}(0,1)$ \\
            \hline

    \end{tabular}
\end{table}
\begin{table}
\caption{The weak mutually unbiased bases in $H_{15}$ and their factorizations in terms of the
mutually unbiased bases $\ket{B_j^{(1)};\overline m_1}$ in $H_{3}$, and  $\ket{B_j^{(2)};\overline m_2}$ in $H_5$, according to Eq.(\ref{S}).}
\centering
\vspace{12pt}
    \begin{tabular}{|r|c|c|c|}
    \hline
       $\ket{{\mathfrak B}_j;m}$&$\ket{X(\kappa,\lambda|\mu,\nu);m}$&$\ket{B_j^{(1)};\overline m_1}$&$\ket{B_j^{(2)};\overline m_2}$ \\
      \hline       
      $\ket{{\mathfrak B}_1;m}$&$\ket{X;m}$&$\ket{X^{(1)};\overline m_1}$&$\ket{X^{(2)};\overline m_2}$ \\
            \hline
            
      $\ket{{\mathfrak B}_2;m}$&$\ket{X(10,12|12,10);m}$&$\ket{X^{(1)};\overline m_1}$&$\ket{X^{(2)}(0,1|-1,0);\overline m_2}$ \\
            \hline
            
      $\ket{{\mathfrak B}_3;m}$&$\ket{X(10,12|12,4);m}$&$\ket{X^{(1)};\overline m_1}$&$\ket{X^{(2)}(0,1|-1,-1);\overline m_2}$ \\
            \hline
            
      $\ket{{\mathfrak B}_4;m}$&$\ket{X(10,12|12,13);m}$&$\ket{X^{(1)};\overline m_1}$&$\ket{X^{(2)}(0,1|-1,-2);\overline m_2}$ \\
            \hline
            
      $\ket{{\mathfrak B}_5;m}$&$\ket{X(10,12|12,7);m}$&$\ket{X^{(1)};\overline m_1}$&$\ket{X^{(2)}(0,1|-1,-3);\overline m_2}$ \\
            \hline
            
      $\ket{{\mathfrak B}_6;m}$&$\ket{X(10,12|12,1);m}$&$\ket{X^{(1)};\overline m_1}$&$\ket{X^{(2)}(0,1|-1,-4);\overline m_2}$ \\
            \hline
            
      $\ket{{\mathfrak B}_7;m}$&$\ket{X(6,5|10,6);m}$&$\ket{X^{(1)}(0,1|-1,0);\overline m_1}$&$\ket{X^{(2)};\overline m_2}$ \\
            \hline
            
      $\ket{{\mathfrak B}_8;m}$&$\ket{X(6,5|10,11);m}$&$\ket{X^{(1)}(0,1|-1,-1);\overline m_1}$&$\ket{X^{(2)};\overline m_2}$ \\
            \hline
            
      $\ket{{\mathfrak B}_9;m}$&$\ket{X(6,5|10,1);m}$&$\ket{X^{(1)}(0,1|-1,-2);\overline m_1}$&$\ket{X^{(2)};\overline m_2}$ \\
            \hline
            
      $\ket{{\mathfrak B}_{10};m}$&$\ket{X(0,2|7,0);m}$&$\ket{X^{(1)}(0,1|-1,0);\overline m_1}$&$\ket{X^{(2)}(0,1|-1,0);\overline m_2}$ \\
            \hline
            
      $\ket{{\mathfrak B}_{11};m}$&$\ket{X(0,2|7,9);m}$&$\ket{X^{(1)}(0,1|-1,0);\overline m_2}$&$\ket{X^{(2)}(0,1|-1,-1);\overline m_2}$ \\
            \hline
            
      $\ket{{\mathfrak B}_{12};m}$&$\ket{X(0,2|7,3);m}$&$\ket{X^{(1)}(0,1|-1,0);\overline m_1}$&$\ket{X^{(2)}(0,1|-1,-2);\overline m_2}$ \\
            \hline
            
      $\ket{{\mathfrak B}_{13};m}$&$\ket{X(0,2|7,12);m}$&$\ket{X^{(1)}(0,1|-1,0);\overline m_1}$&$\ket{X^{(2)}(0,1|-1,-3);\overline m_2}$ \\
            \hline
            
      $\ket{{\mathfrak B}_{14};m}$&$\ket{X(0,2|7,6);m}$&$\ket{X^{(1)}(0,1|-1,0);\overline m_1}$&$\ket{X^{(2)}(0,1|-1,-4);\overline m_2}$ \\
            \hline
            
      $\ket{{\mathfrak B}_{15};m}$&$\ket{X(0,2|7,5);m}$&$\ket{X^{(1)}(0,1|-1,-1);\overline m_1}$&$\ket{X^{(2)}(0,1|-1,0);\overline m_2}$ \\
            \hline
            
      $\ket{{\mathfrak B}_{16};m}$&$\ket{X(0,2|7,14);m}$&$\ket{X^{(1)}(0,1|-1,-1);\overline m_1}$&$\ket{X^{(2)}(0,1|-1,-1);\overline m_2}$ \\
            \hline
            
      $\ket{{\mathfrak B}_{17};m}$&$\ket{X(0,2|7,8);m}$&$\ket{X^{(1)}(0,1|-1,-1);\overline m_1}$&$\ket{X^{(2)}(0,1|-1,-2);\overline m_2}$ \\
            \hline
            
      $\ket{{\mathfrak B}_{18};m}$&$\ket{X(0,2|7,2);m}$&$\ket{X^{(1)}(0,1|-1,-1);\overline m_1}$&$\ket{X^{(2)}(0,1|-1,-3);\overline m_2}$ \\
            \hline
            
      $\ket{{\mathfrak B}_{19};m}$&$\ket{X(0,2|7,11);m}$&$\ket{X^{(1)}(0,1|-1,-1);\overline m_1}$&$\ket{X^{(2)}(0,1|-1,-4);\overline m_2}$ \\
            \hline
            
      $\ket{{\mathfrak B}_{20};m}$&$\ket{X(0,2|7,10);m}$&$\ket{X^{(1)}(0,1|-1,-2);\overline m_1}$&$\ket{X^{(2)}(0,1|-1,0);\overline m_2}$ \\
            \hline
            
      $\ket{{\mathfrak B}_{21};m}$&$\ket{X(0,2|7,4);m}$&$\ket{X^{(1)}(0,1|-1,-2);\overline m_1}$&$\ket{X^{(2)}(0,1|-1,-1);\overline m_2}$ \\
            \hline
            
      $\ket{{\mathfrak B}_{22};m}$&$\ket{X(0,2|7,13);m}$&$\ket{X^{(1)}(0,1|-1,-2);\overline m_1}$&$\ket{X^{(2)}(0,1|-1,-2);\overline m_2}$ \\
            \hline
            
      $\ket{{\mathfrak B}_{23};m}$&$\ket{X(0,2|7,7);m}$&$\ket{X^{(1)}(0,1|-1,-2);\overline m_1}$&$\ket{X^{(2)}(0,1|-1,-3);\overline m_2}$ \\
            \hline
            
      $\ket{{\mathfrak B}_{24};m}$&$\ket{X(0,2|7,1);m}$&$\ket{X^{(1)}(0,1|-1,-2);\overline m_1}$&$\ket{X^{(2)}(0,1|-1,-4);\overline m_2}$ \\
            \hline

    \end{tabular}
\end{table}

\begin{table}[htbp]
\caption{A partition of the set of the maximal lines through the origin in ${\mathbb Z}(15) \times {\mathbb Z}(15)$. All the lines in the set ${\cal S}_j$ (i.e., in the same column) intersect only at the origin.}
\centering
\vspace{12pt}
    \begin{tabular}{|r|c|c|c|c|c|}
    \hline
       ${\cal S}_0$&${\cal S}_1$&${\cal S}_2$&${\cal S}_3$&${\cal S}_4$&${\cal S}_5$ \\
    \hline
       ${\cal L}_1$&${\cal L}_2$&${\cal L}_3$&${\cal L}_4$&${\cal L}_5$&${\cal L}_6$ \\
    \hline
       ${\cal L}_{10}$&${\cal L}_{11}$&${\cal L}_{12}$&${\cal L}_9$&${\cal L}_8$&${\cal L}_7$ \\
    \hline
       ${\cal L}_{16}$&${\cal L}_{17}$&${\cal L}_{18}$&${\cal L}_{13}$&${\cal L}_{14}$&${\cal L}_{15}$ \\
    \hline
       ${\cal L}_{22}$&${\cal L}_{23}$&${\cal L}_{24}$&${\cal L}_{19}$&${\cal L}_{20}$&${\cal L}_{21}$ \\
    \hline
    
    \end{tabular}
\end{table}

\begin{table}[htbp]
\caption{A partition of the set of the weak mutually unbiased bases (for the case $d=15$). 
All the bases in the set ${\cal T}_j$ (i.e., in the same column) are mutually unbiased bases.}
\centering
\vspace{12pt}
    \begin{tabular}{|r|c|c|c|c|c|}
    \hline
       ${\cal T}_0$&${\cal T}_1$&${\cal T}_2$&${\cal T}_3$&${\cal T}_4$&${\cal T}_5$ \\
    \hline
       $\ket{{\mathfrak B}_1;m}$&$\ket{{\mathfrak B}_2;m}$&$\ket{{\mathfrak B}_3;m}$&$\ket{{\mathfrak B}_4;m}$&$\ket{{\mathfrak B}_5;m}$&$\ket{{\mathfrak B}_6;m}$ \\
    \hline
       $\ket{{\mathfrak B}_{10};m}$&$\ket{{\mathfrak B}_{11};m}$&$\ket{{\mathfrak B}_{12};m}$&$\ket{{\mathfrak B}_9;m}$&$\ket{{\mathfrak B}_8;m}$&$\ket{{\mathfrak B}_7;m}$ \\
    \hline
       $\ket{{\mathfrak B}_{16};m}$&$\ket{{\mathfrak B}_{17};m}$&$\ket{{\mathfrak B}_{18};m}$&$\ket{{\mathfrak B}_{13};m}$&$\ket{{\mathfrak B}_{14};m}$&$\ket{{\mathfrak B}_{15};m}$ \\
    \hline
       $\ket{{\mathfrak B}_{22};m}$&$\ket{{\mathfrak B}_{23};m}$&$\ket{{\mathfrak B}_{24};m}$&$\ket{{\mathfrak B}_{19};m}$&$\ket{{\mathfrak B}_{20};m}$&$\ket{{\mathfrak B}_{21};m}$ \\
    \hline
    
    \end{tabular}
\end{table}

\enddocument